\documentclass[12pt]{iopart}
\expandafter\let\csname equation*\endcsname\relax
\expandafter\let\csname endequation*\endcsname\relax

\usepackage{amsmath,amsfonts,amssymb,amsthm}
\usepackage{graphics,graphicx}
\usepackage{color,colordvi}
\usepackage{hyperref}
\usepackage{enumerate}
\usepackage{refcheck}

\def\softd{{\leavevmode\setbox1=\hbox{d}%
\hbox to 1.05\wd1{d\kern-0.4ex{\char039}\hss}}}
\def\softt{{\leavevmode\setbox1=\hbox{t}%
\hbox to \wd1{t\kern-0.6ex{\char039}\hss}}}

\newtheorem{theorem}{Theorem}[section]
\newtheorem{proposition}[theorem]{Proposition}

\newtheorem{remark}{Remark}[section]
\newtheorem{remarks}{Remarks}[section]

\begin{document}

\title[Spiral waveguides spectra]
{Spectral properties of spiral-shaped quantum waveguides}

\author{Pavel Exner$^{1,2}$ and Milo\v{s} Tater$^1$}
\address{1) Nuclear Physics Institute, Czech Academy of Sciences,
Hlavn\'{i} 130, \\ 25068 \v{R}e\v{z} near Prague, Czech Republic}
\address{2) Doppler Institute, Czech Technical University, B\v{r}ehov\'{a} 7, 11519 Prague, \\ Czech Republic}
\ead{exner@ujf.cas.cz, tater@ujf.cas.cz}

\begin{abstract}
We investigate properties of a particle confined to a hard-wall spiral-shaped region. As a case study we analyze in detail the Archimedean spiral for which the spectrum above the continuum threshold is absolutely continuous away from the thresholds. The subtle difference between the radial and perpendicular width implies, however, that in contrast to `less curved' waveguides, the discrete spectrum is empty in this case. We also discuss modifications such a multi-arm Archimedean spirals and spiral waveguides with a central cavity; in the latter case bound state already exist if the cavity exceeds a critical size. For more general spiral regions the spectral nature depends on whether they are `expanding' or `shrinking'. The most interesting situation occurs in the asymptotically Archimedean case where the existence of bound states depends on the direction from which the asymptotics is reached.
\end{abstract}

\pacs{03.65.Ge, 03.65Db}

\vspace{2pc} \noindent{\it Keywords}: quantum waveguides, spiral-shaped regions, Dirichlet boundary

%
\submitto{\JPA}

\section{Introduction}
\setcounter{equation}{0}

Investigation of the dynamics of quantum particles confined to tubular regions gave rise to interesting results among which the probably most unexpected one concerned localized states that owed their existence exclusively to the geometry of the confinement \cite{EK15}. This discovery had an influence, both theoretical and experimental, to the (much older and more developed) theory of electromagnetic waveguides \cite{LCM99} from the simple reason that, at least in some circumstances, the corresponding Maxwell equations are well approximated by the same Helmholz equation that one encounters in the quantum case.

This paper is devoted to a class of quantum waveguide structures which so far escaped attention, namely two-dimensional spiral-shaped strips with hard-wall, that is, Dirichlet boundaries. The lack of attention to this problem is somewhat surprising because spiral regions appear often in physics, usually with practical applications in mind. Without aspiring to an exhausting survey, let us mention a few examples. On the quantum side spirals can be used as guides for cold atoms \cite{JLX15} which may help to construct atomic gyroscopes \cite{JLZ16}. Predictably, a number of examples concern electromagnetic or optical systems, for instance \cite{BBP09, CFW15}, sometimes again with practical purpose in mind such as nanoparticle detection \cite{TLY18} or spectrometry \cite{RLB16}, or combinations of spiral guides \cite{CLV14}. Spiral shapes have also been studied for acoustic waveguides \cite{PRB17}.

Spirals involved in all the above examples are finite, however,  we are going to discuss spiral regions of infinite length. This is not only the usual theoretical license, but also another aim of this paper, namely to show that that many spectral properties of such systems have a truly global character. To specify what we have in mind, consider a quantum particle in the plane divided by Dirichlet conditions at concentric circles of radii $r_n=2\pi an,\, n=1,2,\dots\,$, into the family of annular domains with impenetrable boundaries. The spectrum of this system covers the halfline $\big((2a)^{-1},\infty\big)$ being there \emph{dense pure point}, which one can check either directly or using the limit of suitable radially periodic potentials, cf.~\cite{EF07}, and in addition, there is a \emph{discrete spectrum} below $(2a)^{-1}$ which is infinite and accumulates at $(2a)^{-1}$. We will show that the spectral character changes profoundly if the Dirichlet boundary is instead imposed on an Archimedean spiral of the slope $a$, despite the fact that if we observe the two boundaries in a simply connected region sufficiently distant from the center they look very similar.

Archimedean waveguides are used as a case study. Their essential spectrum starts predictably at the lowest transverse Dirichlet eigenvalue, however, it is absolutely continuous away from the transverse thresholds. What is more surprising is that the passage from concentric circles to the spiral region destroys the infinite discrete spectrum mentioned above. We will also consider modifications of this example which involve making a cavity in the center of the region and/or imposing the Dirichlet condition at more than one spiral arm. Not surprisingly, bound states occur provided the central cavity is large enough; we will analyze this example both theoretically and numerically.

After dealing with the Archimedean case we turn to more general spirals. The decisive property for the spectrum is the asymptotic behaviour of the coil width. The extreme cases are when it diverges or shrinks to zero, then the spectrum covers the whole positive halfline or it is purely discrete, respectively. The situation is more interesting when the region is asymptotically Archimedean in the sense that the coil width has a finite and nonzero limit. The character of this asymptotics then determines whether a discrete spectrum is present in the gap between zero and the continuum threshold. We illustrate this effect on the regions the boundary of which interpolates between the Fermat and Archimedean spirals, adding again a numerical analysis of this example. We conclude the paper with a list of open problems.

\section{Case study: Archimedean waveguide}
\label{s:archimedes}
\setcounter{equation}{0}

Let $\Gamma_a$ be the Archimedean spiral in the plane with the slope $a>0$, expressed in terms of the polar coordinates, $\Gamma_a=\{r=a\theta:\, \theta\ge 0\}$, and denote by $\mathcal{C}_a$ its complement, $\mathcal{C}_a := \mathbb{R}^2\setminus\Gamma_a$ which is an open set. The object of our interest is the operator
$$
 H_a = -\Delta^{\mathcal{C}_a}_\mathrm{D},
$$
the Dirichlet Laplacian in $L^2(\mathcal{C}_a)$. In other words, $H_a$ it is the Laplacian in the plane with Dirichlet condition imposed on $\Gamma_a$. Since scaling transformations change $\Gamma_a$ into Archimedean spiral again, with a different slope, the above introduced $H_a$ is unitarily equivalent to $\big(\frac{a'}{a}\big)^2 H_{a'}$ for any $a'>0$. It is thus sufficient to analyze the spectrum for a fixed value of $a$. For {\ae}sthetical reasons we may  choose $a=\frac12$ in which case we drop the the subscript; the general case is restored easily multiplying the length scale quantities by $2a$, energy by $(2a)^{-2}$, etc.

In the polar coordinates our Hilbert space is $L^2((0,\infty)\times[0,2\pi); r\mathrm{d}r\mathrm{d}\theta)$ with Dirichlet condition at $\{(\theta+2\pi n,\theta): n\in\mathbb{N}_0\}$. This can be equivalently written as $L^2(\Omega_a; r\mathrm{d}r\mathrm{d}\theta)$ where $\Omega_a$ is the skewed strip $\Omega_a:= \{(r,\theta):\, r\in (r_\mathrm{min}(\theta),a\theta),\,\theta>0\}$ and $r_\mathrm{min}(\theta):= \max\{0,a(\theta-2\pi)\}$. The Dirichlet condition is imposed at the boundary points\footnote{They are not a part of $\Omega_a$, as we work with the usual convention by which the region where the particle is confined is an open set.} of $\Omega_a$ with $r>0$. As for $r=0$ we note that the boundary of $\mathcal{C}_a$ is not convex there and the spiral end represents an angle $2\pi$, hence the operator domain is
\begin{equation} \label{domain}
 D(H_a) = \mathcal{H}^2(\Omega_a) \cap \mathcal{H}_0^1(\Omega_a) \oplus \mathbb{C}(\psi_\mathrm{sing}),
\end{equation}
cf.~\cite{Ko67}, where
$$
 \psi_\mathrm{sing}(r,\theta) = \chi(r)\,r^{1/2}\sin\frac12\theta
$$
and $\chi$ is a smooth function with compact support not vanishing at $r=0$. In the conventional way we pass to the unitarily equivalent operator $\tilde{H}_a$ on $L^2(\Omega)$ using
$$
 U: L^2(\Omega_a; r\mathrm{d}r\mathrm{d}\theta)\to L^2(\Omega_a), \quad (U\psi)(r,\theta) = r^{1/2} \psi(r,\theta),
$$
which acts as
\begin{equation} \label{tildeH}
 \tilde{H}_af = -\frac{\partial^2 f}{\partial r^2} - \frac{1}{r^2} \frac{\partial^2 f}{\partial\theta^2} - \frac{1}{4r^2}.
\end{equation}
Note that this differential expression is independent of the spiral width parameter $a>0$, the difference is in the curve at which Dirichlet condition is imposed.

We can also write the quadratic form associated with $H_a$ which is
\begin{align*} 
 q_a:\: q_a[\psi] &= \int_0^\infty \int_{r_\mathrm{min}(\theta)}^{a\theta} \Big[ r\Big|\frac{\partial\psi}{\partial r}\Big|^2  + \frac{1}{r}\Big|\frac{\partial\psi}{\partial\theta}\Big|^2 \Big] \mathrm{d}r \mathrm{d}\theta \\
 & =\int_0^\infty \int_{r/a}^{(r+2\pi a)/a} \Big[ r\Big|\frac{\partial\psi}{\partial r}\Big|^2  + \frac{1}{r}\Big|\frac{\partial\psi}{\partial\theta}\Big|^2 \Big] \mathrm{d}\theta \mathrm{d}r \nonumber
\end{align*}
with the domain consisting of function $\psi\in H^1(\Omega_a)$ satisfying Dirichlet condition at the point of $\partial\Omega_a$ with $r>0$ and such that the limit
\begin{equation} \label{singel}
 \lim_{r\to 0+}\, \frac{\psi(r,\theta)}{\sin\textstyle{\frac12}\theta}
\end{equation}
exists and is independent of $\theta$. Note that when we pass from the operator to the quadratic form integrating by parts, the boundary term
$$
 \lim_{r\to 0+}\, \int_0^{2\pi} \overline{\psi}(r,\theta)\,r\frac{\partial\psi}{\partial r}(r,\theta)\, \mathrm{d}\theta
$$
vanishes even for $\psi=\psi_\mathrm{crit} := r^{-1/2}u_\mathrm{crit}$.

\begin{proposition} \label{prop: infspess}
$\inf\sigma_\mathrm{ess}(\tilde{H}_a) \ge (2a)^{-2}$.
\end{proposition}
\begin{proof}
As pointed out above, it is sufficient to consider $a=\frac12$. We cut the skewed strip $\Omega$ into two parts,
$$
 \Omega^\mathrm{N}:= \big\{(r,\theta):\, r\in (\max\{0,\textstyle{\frac12}\theta-\pi\},\textstyle{\frac12}\theta),\,0\le\theta< \theta_\mathrm{N} \big\}
$$
and
$$
 \Omega^\mathrm{N}_\mathrm{c}:= \big\{(r,\theta):\, r\in (\textstyle{\frac12}\theta-\pi,\,\textstyle{\frac12}\theta),\,\theta> \theta_\mathrm{N} \big\}
$$
for some $\theta_0>2\pi$, and consider the operator obtained by imposing additional Neumann condition at $\theta=\theta_\mathrm{N}$. It has the form of a direct sum and estimates our operator from below, $\tilde{H}\ge\tilde{H}^\mathrm{N} \oplus \tilde{H}^\mathrm{N}_\mathrm{c}$. The first component corresponds to a bounded set so that it does not contribute to the essential spectrum. The quadratic form associated with $\tilde{H}^\mathrm{N}_\mathrm{c}$ is
$$
 q^\mathrm{c}[\psi] = \int_{\frac12\theta_\mathrm{N}-\pi}^\infty \int_{\max\{2r,\theta_\mathrm{N}\}}^{2r+2\pi} \Big[ r\Big|\frac{\partial\psi}{\partial r}\Big|^2  + \frac{1}{r}\Big|\frac{\partial\psi}{\partial\theta}\Big|^2 \Big] \mathrm{d}\theta \mathrm{d}r
$$
defined for functions $\psi\in H^1(\Omega^\mathrm{N}_\mathrm{c})$ vanishing for a fixed value of $r$ at $\theta=\max\{2r,\theta_\mathrm{N}\}$ and $\theta=2(r+\pi)$. Consider first the inner integral.  Since Dirichlet condition is imposed on at least one endpoint, we have
$$
 \int_{\max\{2r,\theta_\mathrm{N}\}}^{2r+2\pi} \Big|\frac{\partial\psi(r,\theta)}{\partial\theta}\Big|^2 \mathrm{d}\theta \ge \frac{1}{16} \int_{\max\{2r,\theta_\mathrm{N}\}}^{2r+2\pi} |\psi(r,\theta)|^2 \mathrm{d}\theta
$$
and for $2r\ge\theta_\mathrm{N}$ the coefficient on the right-hand side changes to $\frac14$. For $f\in\mathrm{Dom}(\tilde{H})$ this is equivalent to
$$
 \int_{\frac12\theta_\mathrm{N}-\pi}^\infty \int_{\max\{2r,\theta_\mathrm{N}\}}^{2r+2\pi} \frac{1}{r^2}\,\Big[-\overline{f(r,\theta}) \frac{\partial^2f(r,\theta)}{\partial\theta^2} - \frac{1}{16}\,|f(r,\theta)|^2   \Big]  r\mathrm{d}r\mathrm{d}\theta \ge 0
$$
which implies the following operator estimate,
$$
 \tilde{H}^\mathrm{N}_\mathrm{c} \ge -\frac{\partial^2}{\partial r^2} - \frac{3}{16r^2}.
$$
However, the `vertical' width of $\Omega_0^\mathrm{c}$ is $\pi$ and $r>\frac12\theta_\mathrm{N}-\pi$ where $\theta_\mathrm{N}$ can be chosen arbitrarily large, hence $\inf\sigma_\mathrm{ess}(\tilde{H}) = \inf\sigma_\mathrm{ess}(\tilde{H}^\mathrm{N}_\mathrm{c}) \ge 1$.
\end{proof}

The question about the existence of spectrum below $(2a)^{-2}$ is equivalent to the positivity violation of the shifted quadratic form,
$$
 \psi\mapsto q_a[\psi] - \frac{1}{(2a)^2}\|\psi\|^2.
$$
Since $\psi(r,r/a)=\psi(r,(r+2\pi a)/a)=0$, we have
\begin{equation} \label{poinctheta}
 \int_{r/a}^{(r+2\pi a)/a)} \Big|\frac{\partial\psi(r,\theta)}{\partial\theta}\Big|^2 \mathrm{d}\theta \ge \frac{1}{4} \int_{r/a}^{(r+2\pi a)/a)} |\psi(r,\theta)|^2 \mathrm{d}\theta,
\end{equation}
hence
$$
 q_a[\psi] - \frac{1}{(2a)^2}\|\psi\|^2 \ge p_{(0,\infty)}[\psi],
$$
where
\begin{equation} \label{pform}
 p_{(\alpha,\beta)}[\psi] :=
 \int_\alpha^\beta  \mathrm{d}\theta  \int_{r_\mathrm{min}(\theta)}^{a\theta} \Big[ r\Big|\frac{\partial\psi(r,\theta)}{\partial r}\Big|^2  + \Big(\frac{1}{4r} - \frac{r}{4a^2}\Big)|\psi(r,\theta)|^2 \Big] \mathrm{d}r
\end{equation}
If $\alpha\ge2\pi$, we have $r_\mathrm{min}(\theta) = a(\theta-2\pi)$, and
\begin{align*}
 & \int_{a(\theta-2\pi)}^{a\theta} r\Big|\frac{\partial\psi(r,\theta)}{\partial r}\Big|^2 \mathrm{d}r = \int_{a(\theta-2\pi)}^{a\theta} \Big[ -\overline{f}(r,\theta) \frac{\partial^2 f(r,\theta)}{\partial r^2} - \frac{1}{4r^2}|f(r,\theta)|^2\Big] \mathrm{d}r \\ &\ge \frac{1}{4a^2}\|f\|^2 - \int_{a(\theta-2\pi)}^{a\theta} \frac{1}{4r^2}|f(r,\theta)|^2 \mathrm{d}r = \int_{a(\theta-2\pi)}^{a\theta} \Big(\frac{r}{4a^2} - \frac{1}{4r} \Big)|\psi(r,\theta)|^2 \mathrm{d}r
\end{align*}
because $f(a(\theta-2\pi),\theta)=f(a\theta,\theta)=0$, which implies
\begin{equation} \label{classforb}
 p_{(\alpha,\beta)}[\psi] \ge 0 \qquad \text{for any}\;\; 2\pi\le\alpha <\beta\le\infty.
\end{equation}
On the other hand, to assess $p_{(\alpha,\beta)}[\psi]$ for $(\alpha,\beta)\subset(0,2\pi)$ we have the Dirichlet condition only at $t=a\theta$, not at $r_\mathrm{min}(\theta)=0$, and therefore the factor $\frac{1}{4a^2}$ is the above estimate has to be replaced by $\big(\frac{\pi}{2a\theta}\big)^2 = \big(\frac{\pi}{2r}\big)^2$ which yields the inequality
\begin{equation} \label{estsmall}
 p_{(\alpha,\beta)}[\psi] \ge
 \int_\alpha^\beta  \mathrm{d}\theta  \int_0^{a\theta} \Big(\frac{\pi^2}{4r} - \frac{r}{4a^2}\Big)|\psi(r,\theta)|^2 \mathrm{d}r.
\end{equation}
This means that $p_{(\alpha,\beta)}[\psi]\ge 0$ for $\beta\le\pi$ and the only negative contribution can come from the interval $(\pi,2\pi)$, in particular, that there can be at most a finite number of bound states. In Sec.~\ref{s:cavity} we will present a convincing numerical evidence that the discrete spectrum is in fact empty.

From what we know about curved hard-wall waveguides this conclusion may seem surprising. In order to understand the reason, we look at the problem from a different point of view introducing another parametrization of $\mathcal{C}_a$, this time by locally orthogonal coordinates -- sometimes called \emph{Fermi} or \emph{parallel} -- in the spirit of \cite[Chap.~1]{EK15}. The Cartesian coordinates of the spiral are
$$
 x_1 = a\theta\cos\theta,\quad x_2 =a\theta\sin\theta,
$$
hence the tangent and (inward pointing) normal vectors are
\begin{align*}
t(\theta) &= \frac{1}{\sqrt{1+\theta^2}}\, \big(\cos\theta-\theta\sin\theta, \sin\theta+\theta\cos\theta\big), \\
n(\theta) &= \frac{1}{\sqrt{1+\theta^2}}\, \big(-\sin\theta-\theta\cos\theta, \cos\theta-\theta\sin\theta\big).
\end{align*}
The transverse coordinate $u$ will then measure the distance from $\Gamma_a$,
\begin{align*}
x_1(\theta,u) &= a\theta\cos\theta - \frac{u}{\sqrt{1+\theta^2}}\, \big(\sin\theta+\theta\cos\theta\big), \nonumber \\[-.7em] 
\\[-.7em] \nonumber
x_2(\theta,u) &= a\theta\sin\theta + \frac{u}{\sqrt{1+\theta^2}}\, \big(\cos\theta-\theta\sin\theta\big),
\end{align*}
with $u>0$. A natural counterpart to the variable $u$ is the \emph{arc length} of the spiral given by
\begin{subequations}
\label{archilength}
\begin{equation} \label{arclength}
 s(\theta) = a \int_0^\theta \sqrt{1+\xi^2}\,\mathrm{d}\xi =  \textstyle{\frac12}a \big( \theta \sqrt{1+\theta^2} + \ln(\theta + \sqrt{1+\theta^2})\big)
\end{equation}
which for large values of $\theta$ behaves as
\begin{equation} \label{arclengthas}
 s(\theta) =  \textstyle{\frac12}a \theta^2 + \mathcal{O}(\ln\theta).
\end{equation}
\end{subequations}
While we do not have an explicit expression for the function inverse to \eqref{arclength}, the last relation yields at least its asymptotic behaviour. Another quantity of interest is the \emph{curvature} of the spiral given by
 \begin{subequations}
\label{archicurv}
\begin{equation} \label{curvature}
 \kappa(\theta) =\frac{2+\theta^2}{a(1+\theta^2)^{3/2}} = \frac{1}{a\theta} + \mathcal{O}(\theta^{-2}) \quad \text{as}\;\; \theta\to\infty
\end{equation}
which means that
\begin{equation} \label{curvas}
 \kappa(s) =\frac{1}{ \sqrt{2as}} + \mathcal{O}(s^{-1}) \quad \text{as}\;\; s\to\infty.
\end{equation}
\end{subequations}
With an abuse of notation we will denote the points of $\mathcal{C}_a$ as $x(s,u)$, however, we have to keep in mind that the described parametrization cannot be used globally, as it becomes non-unique for small $\theta$ when the normal to $\Gamma$ fails to cross the previous coil of the spiral; this obviously happens for $\theta<\theta_0$ with some $\theta_0\in (\frac32\pi,2\pi)$.

Nevertheless, we can use it elucidate the properties of $H_a$ that depend on the behaviour of $\Gamma_a$ at large values of $s$. To be specific, we use the decomposition $L^2(\mathcal{C}_a) = L^2(\mathcal{C}_{a,\mathrm{c}}) \oplus L^2(\mathcal{C}_{a,\mathrm{nc}})$, where
\begin{equation} \label{noncomp}
 \mathcal{C}_{a,\mathrm{nc}} =\big\{ x(s,u):\: s>s(2\pi), u\in(0,d(s)) \big\},
\end{equation}
where $d(s)$ is the distance of the point $x(s,0)$ of $\Gamma_a$ from the previous coil of the spiral, $s(2\pi) \approx \frac12 a \times 42.51 \approx 1.077 \times \frac12 a (2\pi)^2$, and $\mathcal{C}_{a,\mathrm{c}}$ is (the interior of) the complement $\mathcal{C}_a \setminus \mathcal{C}_{a,\mathrm{nc}}$. By $H^\mathrm{D}_{a,\mathrm{nc}}$ and $H^\mathrm{N}_{a,\mathrm{nc}}$ we denote the restriction of $H_a$ to $L^2(\mathcal{C}_{a,\mathrm{nc}})$ with the Dirichlet and Neumann condition, respectively, imposed at the perpendicular cut referring to $s=s(2\pi)$.

The question is now about $d(s)$, the range of the variable $u$, or in other words, the transverse width of $\mathcal{C}_a$ at a given $s$. It is given by the intersection of the normal to $\Gamma$ with the previous coil of the spiral. Let us denote the corresponding angle as $\theta_-$, then we have to solve the equations
\begin{align*}
a\theta\cos\theta - \frac{u}{\sqrt{1+\theta^2}}\, \big(\sin\theta+\theta\cos\theta\big) =&\: a\theta_-\cos\theta_-, \\
a\theta\sin\theta + \frac{u}{\sqrt{1+\theta^2}}\, \big(\cos\theta-\theta\sin\theta\big) =&\: a\theta_-\sin\theta_-.
\end{align*}
This yields, in particular, the relations
\begin{align*}
\theta \Big(a - \frac{u}{\sqrt{1+\theta^2}}\Big) =&\: a\theta_-\cos(\theta-\theta_-), \\
- \frac{u}{\sqrt{1+\theta^2}} =&\: a\theta_-\sin(\theta-\theta_-),
\end{align*}
which can be viewed as equations for $u$ and $\theta_-$, and from that we get
\begin{equation} \label{transw}
 \theta^2 \Big(a^2 - \frac{2au}{\sqrt{1+\theta^2}}\Big) + u^2 = a^2\theta_-^2.
\end{equation}
and
\begin{equation} \label{transang}
 \theta\big( 1+\theta_-\sin(\theta-\theta_-)\big) = \theta_-\cos(\theta-\theta_-).
\end{equation}

We can get easily an asymptotically exact lower bound to the width $d(s)$. Abusing again the notation we write it as $d(\theta)$ meaning $d(s(\theta))$. The radial dropped from the point $x(\theta,0)$ towards the coordinate center crosses the previous coil of $\Gamma_a$ at $x(\theta',0)$ where $\theta'=\theta-2\pi$. The spiral slope at this point, i.e. the angle $\beta(\theta')$ between the tangent to $\Gamma_a$ and the tangent the the circle passing through this point is easily found: we have
$$
\cos\beta(\theta') = \frac{\theta'}{\sqrt{1+\theta'^2}},
$$
and since the radial distance between the two coils is $2\pi a$, we get
\begin{equation} \label{transest}
d(\theta) > \frac{2\pi a\theta'}{\sqrt{1+\theta'^2}} = \frac{2\pi a\theta}{\sqrt{1+\theta^2}} \big(1+\mathcal{O}(\theta^{-1})\big),
\end{equation}
in particular, the transverse contribution to the energy is
$$
\frac{\pi^2}{d(\theta)^2} < \frac{1}{(2a)^2}\big(1+(\theta-2\pi)^{-2}\big) = \frac{1}{(2a)^2}\Big(1 + \frac{1}{\theta^2} + \mathcal{O}(\theta^{-3}) \Big)
$$
To see that the bound \eqref{transest} to the solution $u=d(\theta)$ of \eqref{transw} is asymptotically exact, let us put $\theta_-=\theta'+\delta$ and rewrite \eqref{transang} as
$$
0 = F(\eta,\delta):= \big( 1+\eta(2\pi-\delta)\big) (\eta-\sin\delta) - \eta\cos\delta,
$$
where $\eta:= \theta_-^{-1}$ and use the implicit function theorem to solve this equation in the vicinity of $(\eta,\delta)=(0,0)$. The solution exists because $f_\delta = \frac{\partial F}{\partial\delta}$ satisfies $F_\delta(0,0)=1\ne 0$ and we have
\begin{align*}
\frac{\mathrm{d}\delta}{\mathrm{d}\eta}(0) &= - \frac{F_\eta}{F_\delta}(0,0) = 0, \\
\frac{\mathrm{d}^2\delta}{\mathrm{d}\eta^2}(0) &= -\frac{1}{F_\delta^3}\big( F_{\eta\eta}F_\delta^2 -2F_{\eta\delta}F_\eta F_\delta + F_{\delta\delta}F_\eta^2\big)(0,0) = 4\pi,
\end{align*}
so that
$$
\delta(\theta_-) = \frac{2\pi}{\theta_-^2} + \mathcal{O}(\theta_-^{-4}) = \frac{2\pi}{\theta^2} \big(1+ \frac{4\pi}{\theta} + \mathcal{O}(\theta^{-2})\big).
$$
On the other hand, \eqref{transw} is solved by
$$
d(\theta) = \frac{a\theta^2}{\sqrt{1+\theta^2}} - \sqrt{\frac{a^2\theta^4}{1+\theta^2} - a^2(\theta^2-\theta_-^2)}
$$
and substituting $\theta_- = \theta - 2\pi + \delta(\theta)$ we get
$$
d(\theta) = \frac{2\pi a\theta}{\sqrt{1+\theta^2}} \big(1+\mathcal{O}(\theta^{-1})\big).
$$

\bigskip

The coordinates $s,u$ allow us to pass in the standard way \cite[Chap.~1]{EK15} from the the operator $H^\mathrm{D}_{a,\mathrm{nc}}$ and its Neumann counterpart to a unitarily equivalent operator on $L^2(\Sigma_{a,\mathrm{nc}})$, where
$$
 \Sigma_{a,\mathrm{nc}} =\big\{ (s,u):\: s>s(2\pi), u\in(0,d(s)) \big\}
$$
is a semi-infinite strip with the straight `lower' boundary and the varying, but asymptotically constant width, which acts as
\begin{subequations}
\label{straightening}
\begin{equation} \label{straightH}
\hat{H}^\mathrm{D}_{a,\mathrm{nc}}\psi = -\frac{\partial}{\partial s} (1-u\kappa(s))^{-2} \frac{\partial\psi}{\partial s}(s,u) -\frac{\partial^2\psi}{\partial u^2}(s,u) + V(s,u)\psi(s,u),
\end{equation}
where
\begin{equation} \label{effpot}
V(s,u) := -\frac{\kappa(s)^2}{4(1-u\kappa(s))^2} -\frac{u\ddot\kappa(s)}{2(1-u\kappa(s))^3} -\frac54\,\frac{u^2\dot\kappa(s)^2}{(1-u\kappa(s))^4},
\end{equation}
\end{subequations}
with Dirichlet condition at the boundary of $\mathcal{C}_{a,\mathrm{nc}}$, which in $\hat{H}^\mathrm{N}_{a,\mathrm{nc}}$ is replaced by Neumann one at the cut, $s=s(2\pi)$.

The described parametrization allows us to strengthen Proposition~\ref{prop: spess}.
\begin{proposition} \label{prop: spess}
$\sigma_\mathrm{ess}(H_a) = [(2a)^{-2},\infty)$.
\end{proposition}
\begin{proof}
In view of Proposition~\ref{prop: infspess} it is sufficient to check the inclusion $\sigma_\mathrm{ess}(H_a)\supset [(2a)^{-2},\infty)$. Let us first prove the analogous result for the operator $\hat{H}^\mathrm{D}_{a,\mathrm{nc}}$, without loss of generality we can put again $a=\frac12$. We introduce the following family of functions,
\begin{equation} \label{weyltest}
\psi_{k,\lambda}(s,u) := \mu(\lambda s)\,\mathrm{e}^{iks} \sin\frac{\pi u}{d(s)}
\end{equation}
with $k\in\mathbb{R}$, where $\mu\in C_0^\infty(\mathbb{R})$ with $\mathrm{supp}\,\mu \subset (1,2)$, and use Weyl's criterion. For $\lambda\in (0,s(2\pi)^{-1})$ the support of $\psi_{k,\lambda}$ lies in $\Sigma_{a,\mathrm{nc}}$. Since $d(s)=\pi\big( 1 - \frac{1}{8s} + \mathcal{O}(s^{-3/2})\big)$ and $s^{-1}<\lambda$, we have
$$ 
 \|\psi_{k,\lambda}\| = \lambda^{-1/2} \|\mu\| + \mathcal{O}(1)
$$ 
as $\lambda\to 0$. Next we have to express the norm
$$
\big\|(\hat{H}^\mathrm{D}_{a,\mathrm{nc}} - 1 - k^2)\psi_{k,\lambda}\big\|.
$$
Application of \eqref{straightH} to \eqref{weyltest} produces a complicated expression, however, it is sufficient to single out the terms which dominate in the limit $\lambda\to 0$. To this aim, we note that the presence of $\lambda \dot\mu(s)$ yields the factor $\lambda^{1/2}\|\dot\mu\|$ in the norm, and $\lambda^2 \ddot\mu(s)$ gives rise similarly to $\lambda^{3/2}\|\ddot\mu\|$. Furthermore, $\dot d(s)= \frac{1}8 s^{-2}+\mathcal{O}(s^{-5/2}) = \mathcal{O}(\lambda^2)$ and
\begin{equation} \label{effasympt}
V(s,u) = -\frac{1}{4s} + \mathcal{O}(s^{-3/2}) = \mathcal{O}(\lambda).
\end{equation}
It is easy to see that after cancelation of the leading terms we get
$$
\Big(-\frac{\partial^2\psi}{\partial u^2} -1\Big)\psi_{k,\lambda}(s,u) = \big(\textstyle{\frac{1}{4s}} + \mathcal{O}(s^{-3/2})\big)\psi_{k,\lambda}(s,u) = \mathcal{O}(\lambda)\psi_{k,\lambda}(s,u).
$$
The expression
$$
\Big(-\frac{\partial}{\partial s} (1-u\kappa(s))^{-2} \frac{\partial}{\partial s} k^2\Big)\psi_{k,\lambda}(s,u) \mathcal{O}(\lambda)\psi_{k,\lambda}(s,u);
$$
is more complicated but the leading terms containing $k^2$ cancel and in the next order we have a single one, namely $-2ik\lambda \dot\mu(\lambda s)\,\mathrm{e}^{iks} \sin\frac{\pi u}{d(s)}$; combining these observations we infer that
$$ 
 \frac{\big\|(\hat{H}^\mathrm{D}_{a,\mathrm{nc}} - 1 - k^2)\psi_{k,\lambda}\big\|}{\|\psi_{k,\lambda}\|} = \mathcal{O}(\lambda) \quad \text{as}\;\; \lambda\to 0
$$ 
so that $1+k^2 \in \sigma(\hat{H}^\mathrm{D}_{a,\mathrm{nc}})$ for any $k\in\mathbb{R}$, and the same holds for $\hat{H}^\mathrm{N}_{a,\mathrm{nc}}$ because the supports of the functions \eqref{weyltest} are separated from the boundary of $\Sigma_{a,\mathrm{nc}}$ at $s=s(2\pi)$. Furthermore, by the indicated unitary equivalence their preimages in the original coordinates constitute a Weyl sequence of the operator $H_a$, the `full' one as the cut as $s=s(2\pi)$ is again irrelevant from the viewpoint of the essential spectrum. Finally, one can choose as sequence $\{\lambda_n\}_{n=1}^\infty$ in such a way that the supports of different $\psi_{k,\lambda_n}$ do nor overlap, say, by putting $\lambda_n=2^{-3-n}$. Then $\psi_{k,\lambda_n}\to 0$ weakly as $n\to\infty$ which means that $1+k^2 \in \sigma_\mathrm{ess}(H_a)$ concluding thus the proof.
\end{proof}

Another question concerns the nature of the essential spectrum. To answer it, we employ \emph{Mourre method}, cf. \cite{Mo81} or \cite[Chap.~4]{CFKS87}. As usual in waveguides, we exclude from the consideration the family of  transverse thresholds at which the spectral multiplicity changes, that is, the set $\mathcal{T}=\big\{\big(\frac{n}{2a}\big)^2:\: n=1,2,\dots\,\big\}$.

\begin{proposition} 
Let $I$ be an open interval, $I\subset [(2a)^{-2},\infty)\setminus\mathcal{T}$, then the spectrum of $H_a$ in $I$ is purely absolutely continuous.
\end{proposition}
\begin{proof}
In the spirit of Mourre's method, we have to find a suitable conjugate operator $A$. Working with the unitarily equivalent operator $\tilde{H}_a$ on $\Omega_a$, we choose
\begin{equation} \label{conjugate}
A = -\frac{i}{2}\Big( r\frac{\partial}{\partial r} + \frac{\partial}{\partial r}r \Big)
\end{equation}
with the domain consisting of functions from $\mathcal{H}^1(\Omega_a)$ satisfying Dirichlet condition at the boundary of $\Omega_a$ except its part corresponding to $r=0$. This is the generator of the group $\{\mathrm{e}^{itA}:\, t\in\mathbb{R}\}$ of dilations of the skew strip $\Omega_a$ in the direction parallel to the line $r=a\theta$. It is obvious that the scaling $\mathrm{e}^{itA}$ preserves the domain of $\tilde{H}_a$ corresponding to $D(H_a)$ given by \eqref{domain}, and from the self-similarity it follows that the map $t\mapsto \mathrm{e}^{itA}(\tilde{H}_a-i)^{-1}\mathrm{e}^{-itA}$ has the needed regularity. The commutator $[\tilde{H}_a,iA]$ is easily evaluated,
$$
[\tilde{H}_a,iA]f = -2\,\frac{\partial^2 f}{\partial r^2} - \frac{2}{r^2}\,\frac{\partial^2 f}{\partial \theta^2} - \frac{1}{2r^2} f.
$$
The corresponding quadratic form can be estimated as in \eqref{poinctheta} which shows that the contribution of the of the last two terms in nonnegative and we have
\begin{equation} \label{mourre}
E_{\tilde{H}_a}(I)[\tilde{H}_a,iA]E_{\tilde{H}_a}(I) \ge -2\,\frac{\partial^2}{\partial r^2}\, E_{\tilde{H}_a}(I) \ge \frac18\, E_{\tilde{H}_a}(I)
\end{equation}
because, with at least one boundary Dirichlet, one has
$$
\int_{r_\mathrm{min}(\theta)}^{a\theta} \Big| \frac{\partial f(r,\theta)}{\partial r}\Big|^2 \mathrm{d}r \ge \frac{1}{16} \int_{r_\mathrm{min}(\theta)}^{a\theta} |f(r,\theta)|^2 \mathrm{d}r.
$$
Since the bound \eqref{mourre} contains no compact part, there are no embedded eigenvalues and the spectrum of $H_a$ in $I$ is purely absolutely continuous.
\end{proof}

As we hinted already, a reader familiar with the theory of quantum waveguides \cite{EK15} might expect that there a rich discrete spectrum below $(2a)^{-2}$ because the effective curvature-induced potential \eqref{effpot} is attractive, and moreover, of long-range character as we noted in \eqref{effasympt}. It is not the case, however, and the reason is that the width of $\Sigma_{a,\mathrm{nc}}$ equals $(2a)^{-2}$ only asymptotically. In the sense of quadratic forms we have
$$
\hat{H}^\mathrm{D}_{a,\mathrm{nc}} \ge -\frac{\partial}{\partial s} (1-u\kappa(s))^{-2} \frac{\partial}{\partial s} + W(s,u),
$$
where $W(s,u) := \big(\frac{\pi}{d(s)}\big)^2 + V(s,u)$. Since we do not have an analytic expression for the inverse function to $s(\theta)$, it is more practical to write these quantities in terms of the angular variable $\theta$. We can use the implicit function theorem again to find $d(\theta)$ by (machine assisted) solution of equation \eqref{transang} with respect to $\theta$ directly, obtaining
\begin{equation} \label{transen}
\Big(\frac{\pi}{d(s)}\Big)^2 = \frac{1}{4a^2} + \frac{1}{4a^2\theta^2} + \frac{\pi}{2a^2\theta^3} + \frac{\pi^2}{a^2\theta^4} + \frac{\pi(4\pi^2-1)}{4a^2\theta^5} + \mathcal{O}(\theta^{-6}).
\end{equation}
One the other hand, $V(\theta,u)$ can be expressed from the definition relation \eqref{effpot} using $\dot\kappa(s)=a^{-1}(1+\theta^2)^{-1/2}\dot\kappa(\theta)$ and $\ddot\kappa(s)=a^{-2}(1+\theta^2)^{-1}\ddot\kappa(\theta)$; expanding then the curvature and its derivatives we get
$$
V(s,u) = - \frac{1}{4a^2\theta^2} - \frac{u}{2a^2\theta^3} - \frac{a^2+3u^2}{4a^2\theta^4} + \frac{u(7a^2+4u^2)}{4a^5\theta^5} + \mathcal{O}(\theta^{-6}).
$$
This shows that the contributions due to the strip width and from the effective potential are competing: from the above expansions we get
\begin{equation} \label{combination}
W(s,u) = \frac{1}{4a^2} + \frac{\pi a-u}{2a^3\theta^3} + \frac{a^2(4\pi^2-1)-3u^2}{4a^4\theta^4} + \mathcal{O}(\theta^{-5}).
\end{equation}
We see that, apart from the constant corresponding the continuum threshold, the leading terms cancel mutually here, and the transverse mean of the following one vanishes, modulo a correction coming from the difference between $d(\theta)$ and $2\pi a$. These heuristic considerations correspond well to the observation made above about the sign of the contribution to the form $p_a$ from the region with $\theta>2\pi$.

\section{A variation: spiral waveguide with a cavity}
\label{s:cavity}
\setcounter{equation}{0}

Consider now the situation when one `erases' a part of the Dirichlet boundary imposing the condition on the `cut' spiral $\Gamma_{a,\beta}$ for some $\beta>0$, where $\Gamma_{a,\beta}=\{r=a\theta:\, \theta\ge\beta\}$. The particle is thus localized in its complement, $\mathcal{C}_{a,\beta} := \mathbb{R}^2\setminus\Gamma_{a,\beta}$ which is open set, and its Hamiltonian, modulo unimportant physical constants, is
$$
 H_{a,\beta} = -\Delta^{\mathcal{C}_{a,\beta}}_\mathrm{D},
$$
the Dirichlet Laplacian in $L^2(\mathcal{C}_{a,\beta})$. The above results about the essential spectrum reflect the behaviour of the system outside a compact region, hence they are not affected by the presence of the cavity,
$$
 \sigma_\mathrm{ess}(\tilde{H}_{a,\beta}) = [(2a)^{-2}, \infty).
$$
On the other hand, the discrete spectrum is for sure nonempty provided $\beta$ is large enough.

\begin{proposition} \label{prop: disccavity}
There is a critical angle $\beta_1 = 2j_{0,1} \approx 4.805 \approx 1.531\pi$ such that $\sigma_\mathrm{disc}(H_{a,\beta}) \ne\emptyset$ holds for all $\beta>\beta_1$. Furthermore, let $\mathcal{B}=\{\beta_j\}_{j=1}^\infty$ be the sequence
\begin{equation} \label{critbeta}
\mathcal{B} = \big\{2j_{0,1}, 2j_{1,1}, 2j_{1,1}, 2j_{2,1}, 2j_{2,1}, 2j_{0,2}, 2j_{1,2}, 2j_{1,2}, \dots \big\}
\end{equation}
consisting of multiples of the zeros of Bessel functions $J_n$, $\,n=0,1,\dots$, arranged in the ascending order in which the contributions from $J_n$ with $n\ge1,$ are doubled, then for any $\beta>\beta_j$ the operator $H_{a,\beta}$ has at least $j$ eigenvalues, the multiplicity taken into account\footnote{We expect the discrete spectrum to be simple but we it would need a separate analysis to prove this conjecture; we postpone it to a future work.}.
\end{proposition}
\begin{proof}
The result follows by bracketing; we can again put $a=\frac12$. We estimate $H_{1/2,\beta}$ from above by imposing additional Dirichlet condition at a circle that does not intersect with $\Gamma_{1/2,\beta}$. The latter can be chosen in different ways. One possibility is to choose the circle of radius $\frac12\beta$ with the center at the coordinate origin, alternatively we can take the osculation circle to $\Gamma_{1/2,\beta}$ at the point $\theta=\beta$ the radius of which is
$$
 \rho(\beta) := \frac{(1+\beta^2)^{3/2}}{2(2+\beta^2)}.
$$
We want the circle to be as large as possible, from that point of view the first option is preferred for $\beta>\sqrt{\frac12(\sqrt{5}-1)} \approx 0.786$. The claim will be valid if the lowest Dirichlet eigenvalue in this circle, equal to $j^2_{0,1}(\frac12\beta)^{-2}$ where $j_{0,1}$ is the first zero of Bessel function $J_0$, is smaller than $\inf\sigma_\mathrm{ess}(H_{1/2,\beta})=1$ which happens if $\beta>2j_{0,1}$. The second choice would give $\rho(\beta)>j_{0,1}$ leading to a cubic equation for the square of the critical angle which has a single positive root $\beta_1^2 \approx 24.0189$; this yields a slightly worse result, $\beta_1 \approx 4.901 \approx 1.560\pi$. The same bracketing argument yields the critical values of the angle for a larger number of eigenvalues.
\end{proof}

Note that the bracketing bound could be still improved because one can embed into the cavity a slightly larger circle by shifting its center, however, we avoid trying that as the obtained sufficient condition would be rough only as we will see below.

\subsection{Numerical results}

To get a better insight we also analyze the problem numerically using FEM techniques. The numerical analysis here and in the following sections is performed on finite spiral regions obtained by restricting the range of the angular parameter $\theta$ to a finite interval sufficiently large to ensure numerical stability. There are different ways to check that such a cutoff is appropriately placed. One is based on bracketing \cite[Sec.~XIII.15]{RS78} which we already employed in several proofs above: imposing Dirichlet and Neumann condition at the cutoff we squeeze the true eigenvalue between the corresponding pair of the `truncated' ones which allows us to estimate the error margin. One can also watch the eigenfunctions decay in both cases and place the cutoff far enough to make their values there sufficiently small.

There is more than one way to address our problem numerically. One can apply the FEM technique to the Laplacian in the spiral region directly, alternatively one can analyze operator \eqref{tildeH} in the (truncated) skewed strip depicted on Fig.~\ref{skewstrip} which is easier to tackle.  Using both schemes and comparing the results represents a useful reliability check; the second method also allows for the spectral method as an alternative approach.
\begin{figure}[ht]
\centering
\hspace{3em}\includegraphics[scale=.8]{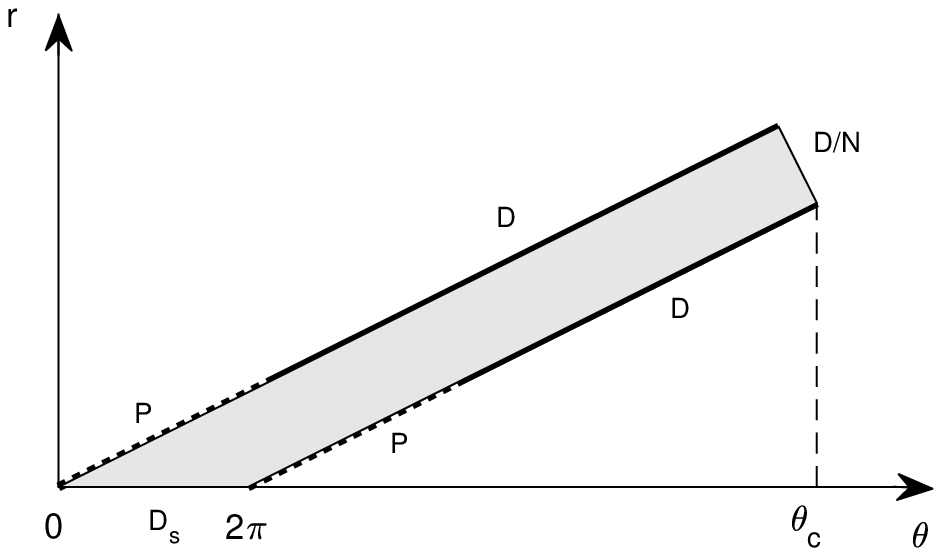}
\caption{A truncated skewed strip for the spiral region with a cavity. For $0<\theta<\beta$ the Dirichlet condition is replaced by the periodic one, at the cutoff both Dirichlet and Neumann conditions are used.}
\label{skewstrip}
\end{figure}

Let us pass to the results. In Fig.~\ref{cavity_ev} we plot the eigenvalues $H_{a,\beta}$ with respect to the cutoff angle $\beta$.
\begin{figure}[ht]
\centering
\includegraphics[scale=1]{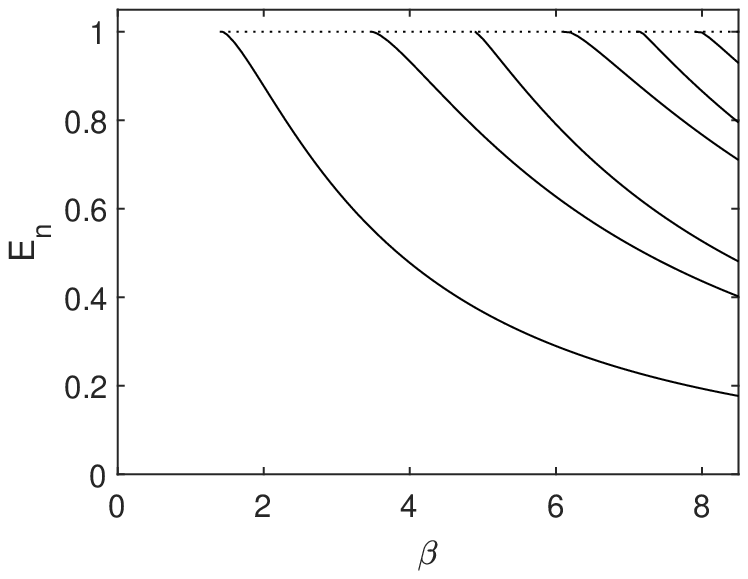}
\caption{Eigenvalues of $H_{1/2,\beta}$ as functions of $\beta$.}
\label{cavity_ev}
\end{figure}
As expected, they are monotonously decreasing functions. We also can identify the critical angle at which the first eigenvalue appears as $\beta_1\approx 1.43 \approx 0.455\pi$. This a substantially smaller value than the sufficient condition of Proposition~\ref{prop: disccavity}, and what is more important, it provides the indication that the discrete spectrum of the `full' Archimedean spiral region is void mentioned in the previous section.

Furthermore, in Fig.~\ref{cavity_ef} we plot the eigenfunctions of $H_{a,21/2}$ indicating their horizontal levels, in particular, their nodal lines.
\begin{figure}[ht]
\centering
\includegraphics[scale=.5]{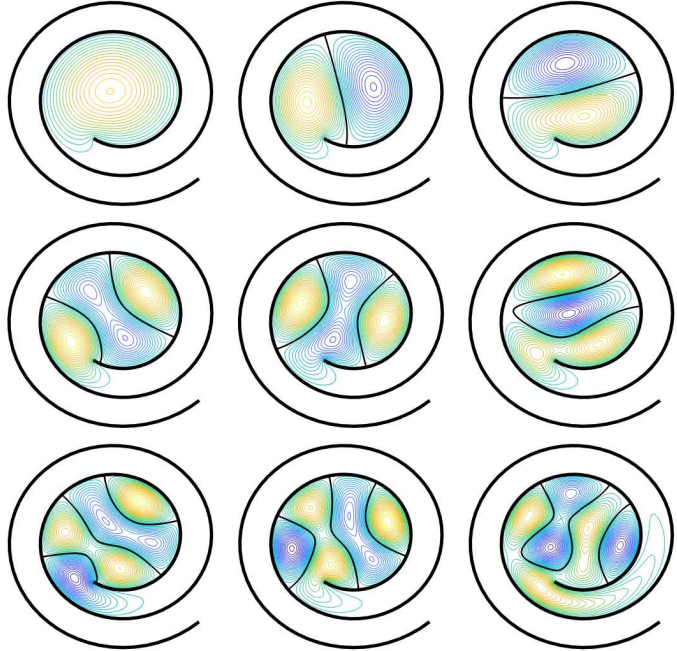}
\caption{The first nine eigenfunctions of $H_{1/2,21/2}$ shown through their horizontal levels (colour online). The corresponding energies are $0.1280$, $0.2969$, $0.3456$, $0.5312$, $05811$, $0.6825$, $0.8266$, $0.8852$, and $0.9768$, respectively.}
\label{cavity_ef}
\end{figure}
The results are in accordance with the Courant nodal domain theorem \cite[Chap.~V.6]{CH53}. We note that the nodal lines are situated in the cavity only which, as well the finiteness of the spectrum, corresponds nicely to the observation expressed by \eqref{classforb} which says that the part of $\mathcal{C}_{a,\beta}$ referring to the angles $\theta>\max\{2\pi,\beta_1\}$ is a classically forbidden zone.

\section{A variation: multi-arm Archimedean waveguide}
\label{s:multi}
\setcounter{equation}{0}

Let $\Gamma_a^m$ be the the union of $m$ Archimedean spirals of the slope $a>0$ with an angular shift, $\Gamma_a^m=\{r=a\big(\theta-\frac{2\pi j}{m}\big):\, \theta\ge \frac{2\pi j}{m},\, j=0,\dots,m-1\}$. As before we denote by $\mathcal{C}_a^m$ its complement, $\mathcal{C}_a^m := \mathbb{R}^2\setminus\Gamma_a^m$, and consider the corresponding Dirichlet Laplacian in $L^2(\mathcal{C}_a^m)$,
$$
 H_a = -\Delta^{\mathcal{C}_a^m}_\mathrm{D},
$$
and note that it again has the scaling property with respect to the slope parameter $a$. Using polar coordinates we work as $L^2(\Omega_a; r\mathrm{d}r\mathrm{d}\theta)$ again but now the skewed strip $\Omega_a$ contains additional Dirichlet boundaries at the lines $r=a\big(\theta+\frac{2\pi}{j}\big),\,j=1,\dots,m-1$. There is a difference, however, coming from the regularity of the boundary. For $m=2$ the set $\mathcal{C}_a^2$ consists of two connected components and has a smooth boundary, for $m\ge 3$ it consists of $m$ connected components separated by the branches of $\Gamma_a^m$, each of them them has an angle at the origin of coordinates which is $\frac{2\pi}{m}$, that is, convex. This means that for any $m\ge 2$ the singular part appearing in \eqref{domain} is missing. Moreover, the operator has a symmetry with respect to the rotation on multiples of the angle $\frac{2\pi}{m}$ which allows us to decompose it into the direct sum of $m$ mutually unitarily equivalent components. It is thus sufficient to investigate operator $\tilde{H}_a^m$ acting as the right-hand side of \eqref{tildeH} in $L^2(\Omega_a^m)$ referring to the skewed strip
$$
 \Omega_a^m:= \big\{(r,\theta):\, r\in (r^m_\mathrm{min}(\theta),a\theta), \theta>0 \big\},
$$
where $r^m_\mathrm{min}(\theta):= \max\big\{0,a\big(\theta-\frac{2\pi}{m}\big)\big\}$, the domain of which is
$$
 D(\tilde{H}_a^m) = \mathcal{H}^2(\Omega_a^m) \cap \mathcal{H}_0^1(\Omega_a^m).
$$
\begin{proposition} 
$\sigma(H_a^m) = \big[\big(\frac{m}{2a}\big)^2,\infty\big)$ for any natural $m\ge 2$. The spectrum is absolutely continuous outside $\mathcal{T}_m=\big\{\big(\frac{mn}{2a}\big)^2\!: n=1,2,\dots\,\big\}$ and its multiplicity is divisible by $m$.
\end{proposition}
\begin{proof}
The claim about the multiplicity follows from the mentioned direct sum decomposition. The arguments used in Sec.~\ref{s:archimedes} to determine the essential spectrum and to prove its absolute continuity outside the thresholds modify easily to the present situation, and it is now also easy to check that there is no spectrum below $\big(\frac{m}{2a}\big)^2$. Indeed, the corresponding quadratic form can be estimated by
$$
 q_a^m[\psi] - \Big(\frac{m}{2a}\Big)^2\|\psi\|^2 \ge p_{(0,\infty)}^m[\psi],
$$
where $p_{(0,\infty)}^m$ is analogous to \eqref{pform} with $\frac{r}{4a^2}$ replaced by $\frac{m^2r}{4a^2}$. As before we can check that
$$
 p_{(\alpha,\beta)}^m[\psi] \ge 0 \qquad \text{for any}\;\; \frac{2\pi}{m}\le\alpha<\beta\le\infty;
$$
on the other hand, the functions from the domain of $p^m_{(\alpha,\beta)}$ now satisfy Dirichlet condition at $r=0$, hence the the second term in the bracket in the analogue of \eqref{estsmall} is four times larger, $\big(\frac{mr}{a}\big)^2$ instead of $\big(\frac{mr}{2a}\big)^2$, and consequently, we have $p_{(0,\infty)}^m[\psi]\ge 0$ for any $\psi\in\mathrm{dom}[q_a^m]$.
\end{proof}

A discrete spectrum can be again generated, of course, if a part of the Dirichlet boundaries is removed. As a simple example, consider a circular cavity centered at the origin, obtained by passing from $\Gamma_a^m$ to $\Gamma_{a,\beta}^m=\{r=a\big(\theta +\frac{2\pi}{j}\big):\, \theta\ge\beta,\, j=0,\dots,m-1\}$. Denoting then the corresponding Dirichlet Laplacian by $H_{a,\beta}^m$ we can modify easily the proof of Proposition~\ref{prop: disccavity} to obtain the following result:
\begin{proposition} \label{prop: m-disccavity}
$\sigma_\mathrm{disc}(H_{a,\beta}^m)\ne\emptyset$ holds for all $m\ge 2$ and $\beta> \frac{2j_{0,1}}{m}$. If $\beta>\frac1m\beta_j$, where $\beta_j$ is the $j$-th element of the sequence \eqref{critbeta}, $H_{a,\beta}^m$ has at least $j$ eigenvalues, the multiplicity taken into account.
\end{proposition}

\subsection{Numerical results}

\begin{figure}[ht]
\centering
\includegraphics[scale=.75]{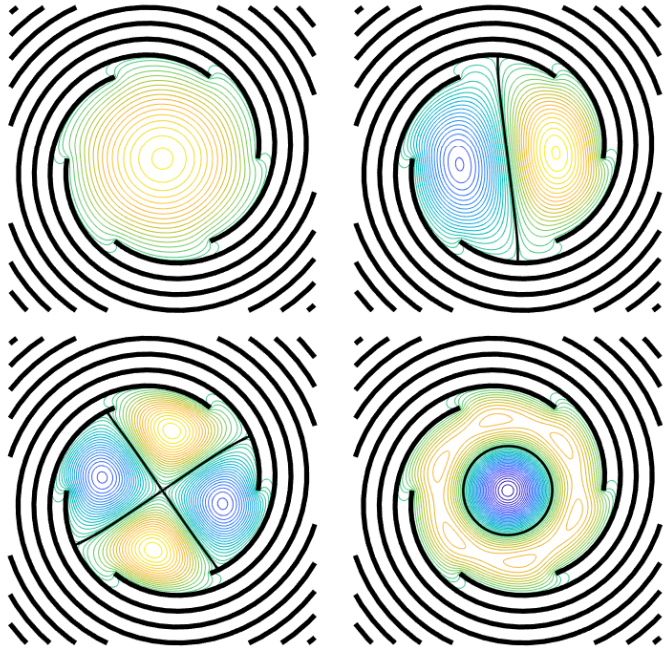}
\caption{The $j$th eigenfunction, $j=1,2,4,6$, of $H^6_{3,2\pi}$, shown through their horizontal levels (colour online). The corresponding energies are $0.1296$, $0.3282$, $0.5871$, and $0.6783$, respectively.}
\label{multicavity_ef}
\end{figure}
To illustrate the results of this part, we plot in Fig.~\ref{multicavity_ef} the eigenfunction $H^m_{a,2\pi}$ for a six-arm spiral region with the central cavity. The picture confirms the expectation that with the growing $m$, the eigenfunctions -- with the possible exception of those corresponding to eigenvalues situated very close to the threshold -- become similar to those of the Dirichlet Laplacian in a disc; it is instructive to compare the nodal lines to those of the single arm region shown in Fig.~\ref{cavity_ef}.

\section{General spirals}
\setcounter{equation}{0}

After looking into the Archimedean case one is naturally interested what can be said about waveguide properties of other spiral shaped regions. A general spiral curve $\Gamma$ is conventionally described in polar coordinates as the family of points $(r(\theta),\theta)$, where $r(\cdot)$ is a given increasing function. Unless specified otherwise, the spirals considered are semi-infinite $r:\mathbb{R}_+ \to \mathbb{R}_+$. In some cases we also consider `fully' infinite spirals for which $r:\mathbb{R} \to \mathbb{R}_+$. The monotonicity of $r$ means that $\Gamma$ does not intersect itself, in other words, the \emph{width function}\footnote{The `inward' coil width at the angle $\theta$ is, of course, $2\pi a(\theta)$; we choose this form with the correspondence to the Archimedean case in mind.}
\begin{equation} \label{coilw}
a:\: a(\theta) = \frac{1}{2\pi}\big(r(\theta)-r(\theta-2\pi)\big).
\end{equation}
is positive for any $\theta\ge2\pi$, or for all $\theta\in\mathbb{R}$ in the fully infinite case. As before we denote $\mathcal{C} := \mathbb{R}^2\setminus\Gamma$ and ask about spectral properties of
$$
 H_r = -\Delta^{\mathcal{C}}_\mathrm{D},
$$
the Dirichlet Laplacian in $L^2(\mathcal{C})$. We restrict our attention to the situation when $r(\cdot)$ is a \emph{$C^1$-smooth function} excluding thus well-known curves such as \emph{Fibonacci spiral}, \emph{spiral of Theodorus}, etc. \cite{wikilist}
\begin{remark} 
{\rm Another modification is represented by \emph{multiarm-arm spirals} generalizing those discussed in Sec.~\ref{s:multi}. Given an $m$-tuple of points $0=\theta_0<\theta_1< \cdots < \theta_{m-1} < 2\pi$ and increasing functions $r_j:\: [\theta_j,\infty)\to\mathbb{R}_+,\: j=0,1,\dots,m-1$, satisfying
$$
a_j(\theta) := \frac{1}{2\pi}\big(r_j(\theta)-r_{j+1}(\theta)\big) >0, \quad \mathcal{H}^2(\Omega) \cap \mathcal{H}_0^1(\Omega),
$$
for all relevant values of $\theta$, we define such a spiral as the family of points $\{(r_j(\theta),\theta):\: j=0,1,\dots,m-1\}$. Note that a two-arm spiral can also be alternatively described by means of a function $r:\mathbb{R}\to\mathbb{R}$ such that $\pm r(\theta)>0$ for \mbox{$\pm\theta>0$} interpreting negative radii as describing vectors rotated by $\pi$. The simplest situation occurs when the system has rotational symmetry, $\theta_j = \frac{2\pi j}{m}$ and $r_j(\theta)=r_0\big(\theta-\frac{2\pi j}{m}\big)$ for \mbox{$0\le j\le m-1$}. The corresponding domain $\mathcal{C}$ then consists of $m$ mutually homothetic components and in view of the Dirichlet conditions at the boundaries one is able to simplify the task by considering one of them only.
}
\end{remark}

One of the most important properties from the waveguide point of view is the asymptotic behaviour of the function of the width function \eqref{coilw}. We call a spiral-shaped region $\mathcal{C}$ \emph{simple} if the function $a(\cdot)$ is monotonous, or in the case of spirals parametrized by $\theta\in\mathbb{R}$ if $a$ is monotonous on each of the halflines $\mathbb{R}_\pm$. A simple $\mathcal{C}$ is called \emph{expanding} and \emph{shrinking} if $a$ is respectively increasing and decreasing for $\theta\ge 0$; these qualifications are labeled as \emph{strict} if $\lim_{\theta\to\infty} a(\theta)=\infty$ and $\lim_{\theta\to\infty} a(\theta)=0$, respectively. For regions referring to multi-arm spirals we use same term if these limit properties apply to all the $a_j,\,j=0,1,\dots,m-1$.

A spiral-shaped region is called \emph{asymptotically Archimedean} if there is an $a_0\in\mathbb{R}$ such that $\lim_{\theta\to\infty} a(\theta) = a_0$, for multi-arm spirals this classification means the existence of finite limits of all the $a_j$. A region $\mathcal{C}$ is obviously unbounded \emph{iff} $\lim_{\theta\to\infty} r(\theta)=\infty$. If the limit is finite, $\lim_{\theta\to\infty} r(\theta)=R$, the closure $\overline{\mathcal{C}}$ is contained in the circle of radius $R$, it may or may not be simply connected as the example of \emph{Simon's jelly roll}, $r(\theta)= \frac34 + \frac{1}{2\pi}\arctan\theta$, shows.

The  Hilbert space can be now written as $L^2(\Omega_r; r\mathrm{d}r\mathrm{d}\theta)$, the skewed strip $\Omega_a$ of the Archimedean case being at that replaced by the region $\Omega_r:= \{(r,\theta):\, r\in (r_\mathrm{min}(\theta),r(\theta)),\,\theta>0\}$, where $r_\mathrm{min}(\theta):= \max\{0,r(\theta-2\pi))\}$ and the Dirichlet condition is imposed at the boundary points with $r>0$, while for $r=0$ we have to add a singular element, so that the domain of $H_r$ is
$$ 
 D(H_r) = \mathcal{H}^2(\Omega_r) \cap \mathcal{H}_0^1(\Omega_r) \oplus \mathbb{C}(\psi_\mathrm{sing})
$$ 
with
$$
 \psi_\mathrm{sing}(r,\theta) = \chi(r)\,r^{1/2}\sin\frac12\theta,
$$
$\chi$ being again a smooth function with a compact support not vanishing at $r=0$. For multi-arm spiral regions we replace a single $\Omega$ by $m$ regions $\Omega_{r,j}:= \{(r,\theta):\, r\in (r_{j,\mathrm{min}}(\theta),r_j(\theta)),\,\theta>0\}$ with $r_{j,\mathrm{min}}(\theta):= \max\{0,r_{j+1}(\theta-\frac{2\pi}{m}))\}$ and the domain of the $j$th component of the operator is $\mathcal{H}^2(\Omega_{r,j}) \cap \mathcal{H}_0^1(\Omega_{r,j})$, $j=0,1,\dots,m-1$\footnote{Unless one of the angles at which the spiral arms meet at the origin is non-convex, of course, in such a case a singular element would be present again.}. Finally, for `fully' infinite spiral regions the operator domain is $D(H_r) = \mathcal{H}^2(\Omega_r) \cap \mathcal{H}_0^1(\Omega_r)$.

As in the particular case discussed above we can get rid of the Jacobian and consider the $\tilde{H_r}$ on $L^2(\Omega_r)$ obtained by means of the unitary transformation
$$
 U: L^2(\Omega_r; r\mathrm{d}r\mathrm{d}\theta)\to L^2(\Omega_r), \quad (U\psi)(r,\theta) = r^{1/2} \psi(r,\theta),
$$
and the action of $\tilde{H}_r$ is again independent of the function $r(\cdot)$,
$$ 
 \tilde{H}_rf = -\frac{\partial^2 f}{\partial r^2} - \frac{1}{r^2} \frac{\partial^2 f}{\partial\theta^2} - \frac{1}{4r^2},
$$ 
and similarly for the multi-arm spiral regions. The associated quadratic form is
\begin{align*} \label{genqform}
 q_r:\: q_r[\psi] &= \int_0^\infty \int_{r_\mathrm{min}(\theta)}^{r(\theta)} \Big[ r\Big|\frac{\partial\psi}{\partial r}\Big|^2  + \frac{1}{r}\Big|\frac{\partial\psi}{\partial\theta}\Big|^2 \Big] \mathrm{d}r \mathrm{d}\theta \\
 & =\int_0^\infty \int_{\theta^{-1}(r)}^{\theta^{-1}(r)+2\pi} \Big[ r\Big|\frac{\partial\psi}{\partial r}\Big|^2  + \frac{1}{r}\Big|\frac{\partial\psi}{\partial\theta}\Big|^2 \Big] \mathrm{d}\theta \mathrm{d}r, \nonumber
\end{align*}
where $\theta^{-1}(\cdot)$ is the pull-back of the function $r(\cdot)$; its domain consists of function $\psi\in H^1(\Omega_r)$ satisfying Dirichlet condition at the point of $\partial\Omega_a$ with the singular element \eqref{singel} added. The modifications for multi-arm and `fully' infinite spirals are again straightforward.

Using function $r$ and its derivatives, we also can express the quantities corresponding to \eqref{arclength} and \eqref{curvature} in the general case; they are
\begin{equation} \label{genarclength}
 s(\theta) = \int_0^\theta \sqrt{\dot{r}(\xi)^2 + r(\xi)^2}\,\mathrm{d}\xi
\end{equation}
and
$$ 
 \kappa(\theta) =\frac{r(\theta)^2 + 2\dot{r}(\theta)^2 - r(\theta)\ddot{r}(\theta)}{(r(\theta)^2 + \dot{r}(\theta)^2)^{3/2}}.
$$ 
In contrast to the Archimedean case, however, it is not always possible to amend the arclength \eqref{genarclength} with the orthogonal coordinate $u$ to parametrize $\mathcal{C}_r$ \emph{globally}, at least with an exception of a finite subset of the plane, by
\begin{align*}
x_1(\theta,u) &= r(\theta)\cos\theta - \frac{u}{\sqrt{\dot{r}(\theta)^2 + r(\theta)^2}}\, \big(\dot{r}(\theta)\sin\theta+r(\theta)\cos\theta\big), \nonumber \\[-.7em] 
\\[-.7em] \nonumber
x_2(\theta,u) &= r(\theta)\sin\theta + \frac{u}{\sqrt{\dot{r}(\theta)^2 + r(\theta)^2}}\, \big(\dot{r}(\theta)\cos\theta-r(\theta)\sin\theta\big).
\end{align*}
The reason is that for strictly expanding spirals the inward normal at a point may not intersect the previous spiral coil; it is easy to check that in the examples of a \emph{logarithmic spiral}, $r(\theta)= a\,\mathrm{e}^{k\theta}$ with $a,k>0$, or \emph{hyperbolic spiral}, $r(\theta)= a\theta^{-1}$, in which cases this happens for all $\theta$ and $\theta^{-1}$, respectively, large enough; note that for the line $n(\theta)$ spanned by the normal to $\Gamma$ at the point $\big(r(\theta),\theta\big)$ we have
$$
 \mathrm{dist}(n(\theta),O) = \frac{r(\theta)\dot{r}(\theta)}{\sqrt{r(\theta)^2 + \dot{r}(\theta)^2}},
$$
where $O$ is the origin of the coordinates. Fortunately, some properties of $H_r$ can be derived without the use of the locally orthogonal system.
\begin{proposition} \label{prop: genspecess}
$\sigma(H_r) = \sigma_\mathrm{ess}(H_r) = [0,\infty)$ holds if $\Gamma$ is simple and strictly expanding.
\end{proposition}
\begin{proof}
Under the assumption one can find a disjoint family $\{\mathcal{D}_n\}\subset \mathcal{C}_r$ of discs with centers at points $x_n\in \mathcal{C}_r$ and $\mathrm{diam}\,\mathcal{D}_n=n$. Then it easy to construct a Weyl sequence, say $\frac1n\,\phi\big(\frac1n |\cdot-x_n|\big)\, \mathrm{e}^{ipx}$, where $\phi\in C_0^\infty(\mathbb{R}_+)$ is nonnegative with $\mathrm{supp}\,\phi\subset[0,2)$ and $\|\phi\|=1$, which shows that $|p|^2\in \sigma_\mathrm{ess}(H_r)$ for any $p\in\mathbb{R}^2$. Since $H_r$ is positive by definition, the claim follows.
\end{proof}

On the other hand, the parallel coordinates can be used in regions generated by shrinking or asymptotically Archimedean spirals provided we exclude a suitable central region as we did in \eqref{noncomp}. Using them, for instance, we obtain a result which is an antipode of Proposition~\ref{prop: genspecess}.
\begin{proposition} \label{prop: genspecdisc}
If $\Gamma$ is simple and strictly shrinking, the spectrum of $H_r$ is purely discrete.
\end{proposition}
\begin{proof}
We divide $\mathcal{C}$ into a precompact $\mathcal{C}_\mathrm{c}$ and $\mathcal{C}_\mathrm{nc} = \big(\mathcal{C}\setminus\mathcal{C}_\mathrm{c}\big)^\mathrm{o}$ in which the parallel coordinates can be introduced, separated by Neumann boundary at some $s_0>0$. It is sufficient to show that restrictions of $H_r$ to the two regions have purely discrete spectra. For the former it follows from the precompacteness of $\mathcal{C}_\mathrm{c}$, the latter is unitarily equivalent to the operator $\hat{H}^\mathrm{D}_\mathrm{nc}$  of the form \eqref{straightening}. Estimating the transverse contribution to the energy from below we get
\begin{equation} \label{shrinkest}
 \hat{H}^\mathrm{D}_\mathrm{nc} \ge -\frac{\partial}{\partial s} (1-u\kappa(s))^{-2} \frac{\partial}{\partial s} +\frac{\pi^2}{d(s)^2} + V(s,u)
\end{equation}
Since $d(s)\to 0$ as $s\to\infty$ holds is a strictly shrinking region, the sum of the two last term explodes in the limit and the absence of the essential spectrum can be checked in the standard way, see e.g. \cite[Thm.~XIII.16]{RS78}.
\end{proof}

\begin{remark} \label{r:jelly roll}
{\rm An example of a strictly shrinking spiral domain is provided by \emph{Fermat spiral}, $r(\theta)^2=b^2\theta$. Furthermore, the above result holds \emph{a fortiori} for regions contained in a bounded part of the plane, $\lim_{\theta\to\infty} r(\theta)<\infty$, where the spectrum would be purely discrete even without the Dirichlet spiral barrier. This illustrates one more time the huge difference between Dirichlet and Neumann Laplacians. Recall an example in which that the Dirichlet barrier is replaced by Neumann one and the spectrum covers the halfline $\mathbb{R}_+$ being absolutely continuous, up to a possible discrete family of embedded eigenvalues of finite multiplicity \cite{Si92}.
}
\end{remark}

\subsection{Numerical results}

As an illustration, we plot in Fig.~\ref{fermat_ef} several eigenfunctions of $H_r$ corresponding to the Fermat spiral region mentioned in the previous remark. By Proposition~\ref{prop: genspecdisc} the spectrum is purely discrete and we see that apart from the central region the eigenfunctions have a quasi-one-dimensional character.
\begin{figure}[ht]
\centering
\includegraphics[scale=.75]{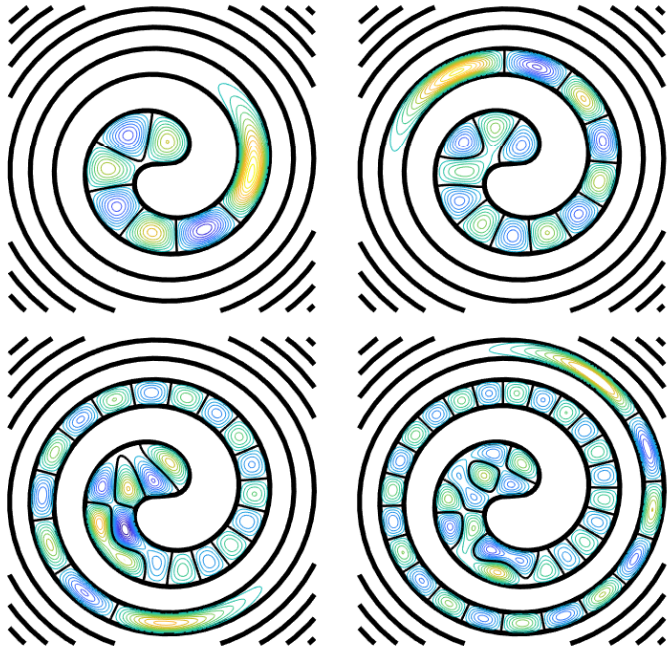}
\caption{Eigenfunctions of the Fermat spiral region with $b=1$, shown through their horizontal levels (colour online), corresponding to the eigenvalues, $E_7=19.5462$, $E_{15}=28.3118$, $E_{27}=38.8062$, and $E_{42}=48.8367$.}
\label{fermat_ef}
\end{figure}
A natural question concerns the distribution of eigenvalues with respect to the increasing energy. The dominant contribution comes from the second term on the right-hand side of \eqref{shrinkest}, the transverse Dirichlet eigenvalue $v(\theta) =\big(\frac{\pi}{d(\theta)}\big)^2$ as a function of the arc-length variable $s$, where this time we have $d(\theta) = r(\theta)-r(\theta-\pi) \approx \frac12\pi b\theta^{-1/2} \big(1+\mathcal{O}(\theta^{-1}(\big)$ as $\theta\to\infty$. Using \eqref{genarclength} we find $s(\theta)=\frac23\,b\,\theta^{3/2} \big(1+\mathcal{O}(\theta^{-1}(\big)$, hence the said leading term of the effective potential is $\approx c\,s^{2/3}$ with $c=\big(\frac{2}{b}\big)^2 \big(\frac{3}{2b}\big)^{2/3}$ giving the estimate of the number of eigenvalues with energy $\le E$ is $N(E) \approx \frac{1}{64}\,b^4\,E^2$ as $E\to\infty$.
\begin{figure}[ht]
\centering
\includegraphics[scale=.75]{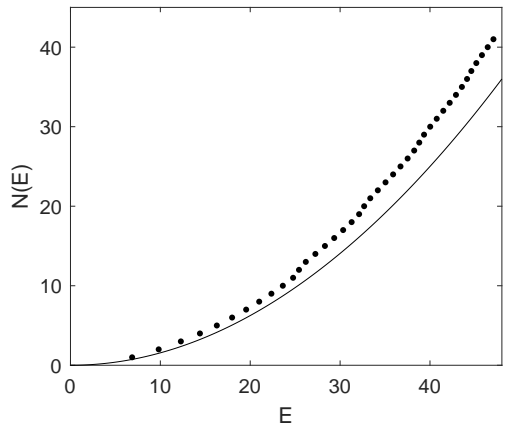}
\caption{The number of eigenvalues in Fermat spiral region \emph{vs.} energy in comparison with the asymptotics taking into account the strip width only.}
\label{fermat_ev}
\end{figure}
This is plotted in Fig.~\ref{fermat_ev} together with the true number of eigenvalues. We see a significant excess which is naturally attributed to the geometry-related effects expressed by first and third terms on the right-hand side of \eqref{shrinkest}. In addition, the plot suggests that there are finer effects in the asymptotics; we leave this to a future work.

\subsection{Asymptically Archimedean regions}

Between the above discussed extremes the situation is more interesting:
\begin{proposition} \label{prop: asymptess}
If the spiral $\Gamma$ is asymptotically Archimedean with $\lim_{\theta\to\infty} a(\theta) = a_0$, we have $\sigma_\mathrm{ess}(H_r) = [(2a_0)^{-2},\infty)$. In the case of a multi-arm region with $\lim_{\theta\to\infty} a_j(\theta) = a_{0,j}$, the essential spectrum is $[(2\overline{a})^{-2},\infty)$, where $\overline{a}:= \max_{0\le j\le m-1} a_{0,j}$.
\end{proposition}
\begin{proof}
Mimicking the proof of Proposition~\ref{prop: spess} we can check that any $\lambda \ge (2a_0)^{-2}$ belongs to the essential spectrum. To complete the proof, we note that
$$
 \inf \sigma_\mathrm{ess}(H_r) \ge \inf \sigma_\mathrm{ess}(\hat{H}^\mathrm{D}_\mathrm{nc}) \ge \inf\big\{\textstyle{\frac{\pi^2}{d(s)^2}} + V(s,u):\: s>s_0,\, 0<u<d(s)\big\}.
$$
The point $s_0$ where the Neumann segment separating $\mathcal{C}_\mathrm{c}$ and $\mathcal{C}_\mathrm{nc}$ is placed can be chosen arbitrarily large, and since the expression at the right-hand side tends to $(2a_0)^{-2}$ as $s_0\to \infty$ we get the sought result. The argument in the multi-arm case is analogous.
\end{proof}

The question about the discrete spectrum of asymptotically Archi\-medean spiral regions is more subtle. There is no doubt that, similarly to Propositions~\ref{prop: disccavity} and \ref{prop: m-disccavity}, localized states arise when $\mathcal{C}$ would be locally wide enough. However, the discrete spectrum, and a rich one, may also come from the asymptotic behaviour of the spiral width provided that $r(\cdot)$ approaches the linear asymptote \emph{from below}. To be specific, consider the spiral
\begin{equation} \label{asspiral}
 r(\theta) =a_0\theta + b_0 - \rho(\theta)
\end{equation}
with the polar width function
$$
 a(\theta) =a_0 + a_1(\theta),\quad a_1(\theta):= \frac{1}{2\pi}\big(\rho(\theta-2\pi) - \rho(\theta)\big),
$$
where $\rho(\cdot)$ is a positive function such that $\lim_{\theta\to\infty} \rho(\theta) = 0$; for the sake of definiteness we restrict our attention to functions satisfying
\begin{equation} \label{derdecay}
 \dot\rho(\theta) = -\frac{c}{\theta^\gamma} + \mathcal{O}(\theta^{-\gamma-1}) \quad\text{as}\quad \theta\to\infty
\end{equation}
with $1<\gamma<3$, so the width function variation is $a_1(\theta) = \frac{2\pi c}{\theta^\gamma} + \mathcal{O}(\theta^{-\gamma-1})$. The restriction $\gamma>1$ means that the spiral is not too `inflated', there is a finite bound to the distance between the respective coils of the spirals with different $c$'s, including the one with $c=0$. Under assumption \eqref{derdecay} we have $a_1(\theta) = \frac{2\pi c}{\theta^\gamma} + \mathcal{O}(\theta^{-\gamma-1})$, hence the length-type quantities such as the orthogonal distance $d(\theta)$ scale by the factor
$$
 1 + \frac{a_1(\theta)}{a_0} = 1 + \frac{2\pi c}{a_0\theta^\gamma} + \mathcal{O}(\theta^{-\gamma-1}),
$$
and energy-type quantities such as $\big(\frac{\pi}{d(s)}\big)^2$ by
\begin{equation} \label{transscale}
 \Big(1 + \frac{a_1(\theta)}{a_0}\Big)^{-2} = 1 - \frac{4\pi c}{a_0\theta^\gamma} + \mathcal{O}(\theta^{-\gamma-1}),
\end{equation}
hence the sum \eqref{combination} of the energy of the lowest transverse mode and the effective potential \eqref{effpot} behaves as
\begin{equation} \label{inflatsum}
 W(s,u) = \frac{1}{4a_0^2} - \frac{\pi c}{a_0^3 \theta^\gamma} + \mathcal{O}(\theta^{-\gamma'}),
\end{equation}
where $\gamma':= \min(\gamma+1,3)$. This allows us to make the following claim:
\begin{proposition} \label{prop: infspec}
Under the assumption \eqref{derdecay}, $\#\sigma_\mathrm{disc}(H_r)=\infty$ holds for any $c>0$.
\end{proposition}
\begin{proof}
We use a variational argument with trial functions $\psi_\lambda=\psi_{0,\lambda}$ given by \eqref{weyltest}; the aim is to show that
\begin{align} \label{shifted}
p[\psi_\lambda] :=&\: \big(\psi_\lambda, \tilde{H}^\mathrm{D}_\mathrm{nc}\psi_\lambda\big) - \frac{1}{4a_0^2}\,\|\psi_\lambda\|^2 \\
=&\: \big(\psi_\lambda, -\partial_s(1-u\kappa(s))^{-2}\partial_s \psi_\lambda + \big( \textstyle{\frac{\pi^2}{d(s)^2}} + V(s,u) - \textstyle{\frac{1}{4a_0^2}} \big) \psi_\lambda\big) \nonumber
\end{align}
can be made negative; as in the proof of Proposition~\ref{prop: spess} $\lambda$ is to be chosen small enough to ensure that the support of $\psi_\lambda$ lies in the region where the parallel-coordinate parametrization applies. The right-hand side is the sum of two terms; the first one can be through integration by parts rewritten as
$$
 \lambda^2 \int_{\lambda^{-1}}^{2\lambda^{-1}} \int_0^{d(s)} \big(1-u\kappa(s))^{-2}\big)\, |\dot\mu(\lambda s)|^2\, \sin^2\!\textstyle{\frac{\pi u}{d(s)}}\, \mathrm{d}u\mathrm{d}s
$$
and since $\kappa(s)\to 0$ and $d(s)\to a_0$ holds as $s\to\infty$, the expression can be for small enough $\lambda$ estimated, say, by $\lambda\,\frac{4\pi}{a_0}\, \|\dot\mu\|^2$. The second term on the right-hand side of \eqref{shifted} can be with the help \eqref{inflatsum} written as
$$
 \int_{\lambda^{-1}}^{2\lambda^{-1}} \int_0^{d(s)} \big(W(s,u) - \textstyle{\frac{1}{4a_0^2}} \big) |\psi_\lambda(s,u)|^2\, \mathrm{d}u\mathrm{d}s.
$$
To proceed we have to relate the arc length \eqref{genarclength} to the angular variable It is easy to see that for the radius given by \eqref{asspiral} the change against \eqref{archilength} shows in the error term only, we have $s(\theta)=\frac12 a_0\theta^2+\mathcal{O}(\theta)$, and conversely, $\theta= \sqrt{\frac{2s}{a_0}} +\mathcal{O}(1)$ holds for large $s$, so \eqref{inflatsum} becomes
\begin{equation} \label{inflatsum-s}
 W(s,u) = -\frac{\pi c}{a_0^3} \Big(\frac{a_0}{2s}\Big)^{\gamma/2} + \mathcal{O}(s^{-\gamma'/2}).
\end{equation}
The leading term is negative, independent of $u$ and increasing with respect to $s$, hence its maximum on the support of $\psi_\lambda$ is reached at $s=2\lambda^{-1}$ and we can estimate the expression in question as
$$
 -\Big(\frac{\pi c}{a_0^3} \Big(\frac{a_0\lambda}{4}\Big)^{\gamma/2} + \mathcal{O}(\lambda^{\gamma'/2})\Big) \int_{\lambda^{-1}}^{2\lambda^{-1}} \frac{2\pi}{d(s)}\, |\mu(\lambda s)|^2 \mathrm{d}s,
$$
hence we finally get
$$
 p[\psi_\lambda] < \lambda\,\frac{4\pi}{a_0}\, \|\dot\mu\|^2 - \Big(\frac{4\pi^2 c}{a_0^4} \Big(\frac{a_0}{4}\Big)^{\gamma/2}\,\lambda^{(\gamma-2)/2}+ \mathcal{O}\big(\lambda^{(\gamma'-2)/2}\big)\Big) \|\mu\|^2,
$$
where the right-hand side is negative for all $\lambda$ small enough. Finally, since the support of $\mu$ is compact, one can choose a sequence $\{\lambda_n\}$ such that $\lambda_n\to 0$ as $n\to\infty$ and the supports of $\psi_{\lambda_n}$ are mutually disjoint which means that the discrete spectrum of $H_r$ is infinite, accumulating at the threshold $(2a_0)^{-2}$.
\end{proof}

\section{Fermat meets Archimedes}
\setcounter{equation}{0}

We have mentioned that Fermat spiral gives rise to the region with a purely discrete spectrum. Something different happens if we consider an interpolation between it and an Archimedean spiral, in the simplest case described parametrically as
\begin{equation} \label{fermarch}
 r(\theta) = a\sqrt{\theta\big(\theta+\textstyle{\frac{b^2}{a^2}}\big)}, \quad a,b>0,
\end{equation}
with the asymptotic behaviour
\begin{align*}
r(\theta) &= b\sqrt{\theta} + \frac{a^2}{2b}\theta^{3/2} + \mathcal{O}(\theta^{5/2}), \\
r(\theta) &= a\theta + \frac{b^2}{2a} + \mathcal{O}(\theta^{-1}),
\end{align*}
for $\theta\to 0+$ and $\theta\to\infty$, respectively. The Fermat spiral is conventionally considered as a two-arm one dividing the plane into a pair of mutually homothetic regions, the interpolating one would thus coincide asymptotically with the two-arm Archimedean spiral. In view of Proposition~\ref{prop: asymptess} the essential spectrum of the corresponding $H_r$ is $[a^{-2},\infty)$. What concerns its discrete spectrum, taking the asymptotic expansion of \eqref{fermarch} two terms further, we can write $r(\theta)$ in the form \eqref{asspiral} with $b_0=\frac{b^2}{2a}$ and
\begin{equation} \label{fermarho}
 \rho(\theta) = \frac{b^4}{8a^3\theta} - \frac{3b^6}{16a^5\theta^2} + \mathcal{O}(\theta^{-3})
\end{equation}
which means that the present interpolating region satisfies assumption \eqref{derdecay} with $c=\frac{b^4}{8a^3}$ and $\gamma=2$. Proposition~\ref{prop: infspec} then says that the operator $H_r$ has an \emph{infinite discrete spectrum} in $(0,a^{-2})$ accumulating at the threshold. Moreover, one can also specify the corresponding accumulation rate: the one-dimensional effective potential \eqref{inflatsum-s} is in this case $\frac{\pi b^4}{16a^5}\, s^{-1} + \mathcal{O}(s^{-3/2})$, with the leading term of Coulomb type, which shows that the number of eigenvalues below $a^{-2}-E$ behaves as
\begin{equation} \label{accurate}
 \mathcal{N}_{a^{-2}-E}(H_r) = \frac{\pi b^4}{32a^5}\,\frac{1}{\sqrt{E}} + o(E^{-1/2}) \quad\text{if}\quad E\to 0+
\end{equation}

\begin{remarks}
{\rm (a) It is worth noting that while the experimentalists sometimes label their spirals explicitly as Archimedean \cite{CFW15, CLV14, JLX15}, in fact they are not. The reason is that they are produced by coiling fibers of a fixed cross section, hence their transverse width is constant instead of changing with the angle as it would for the true Archimedean spiral. Comparing relations \eqref{transen}, \eqref{transscale}, and \eqref{fermarho} we see that such waveguides behave asymptotically rather as the current interpolation with $\frac{b}{a}=(2\pi)^{-1/4}\approx 0.632$. \\[.2em]
(b) Let us finally note that different asymptotically Archimedean regions may exhibit very different spectral behaviour. Consider, for instance, \emph{involute of a circle}, $r(\theta)= a\sqrt{1+\theta^2}$, where the asymptotic behaviour of $r(\theta)$ has the form \eqref{asspiral} with $b_0=0$ and $\rho(\theta)<0$ which means that Proposition~\ref{prop: infspec} does not apply; the polar width of the region is smaller than that of its Archimedean asymptote acting thus against the existence of bound states. A similar conclusion one can make for \emph{Atzema spiral} \cite{wikilist}, for which $r(t)=a(t+t^{-1})$, where $t>t_0$ for a suitable $t_0$ to avoid self-intersections. The parameter in this case is \emph{not} the polar angle, but $\theta = t - \frac{\pi}{2} + \mathcal{O}(t^{-1})$, hence we have $b_0=0$ and $\rho(\theta)<0$ again. On the other hand, modifying the last example slightly by choosing  $r(\theta)= a(\theta-\theta^{-1})$ with $\theta>1$, such a spiral obviously satisfies \eqref{asspiral} with $b_0=0$ and $\gamma=2$ so that the corresponding region gives rise to an infinite discrete spectrum below $(2a)^{-1}$ accumulating in a way similar to \eqref{accurate}.}
\end{remarks}

\subsection{Numerical results}

Let us finally show some numerical results about the interpolation spiral region discussed here. By Proposition~\ref{prop: infspec} the discrete spectrum is infinite. In Fig.~\ref{interpol_ev} we plot the lowest nine eigenvalues of the corresponding operator $H_r$ with respect to the parameter $b$ that controls the interpolation; we choose $a=1$ to fix the essential spectrum.
\begin{figure}[ht]
\centering
\includegraphics[scale=1]{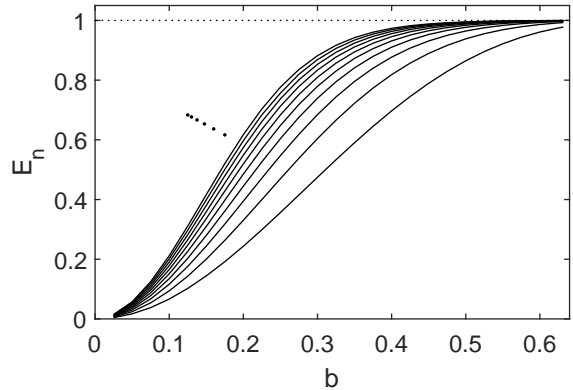}
\caption{The lowest eigenvalues of interpolating region as functions of $b$.}
\label{interpol_ev}
\end{figure}
As expected the ground state is close to the continuum threshold for (sufficiently) large values of $b$ and the whole discrete spectrum disappears in the limit $b\to\infty$, while for small $b$ the region has a large bulge in the center and the spectral bottom drops to appropriately low values. We also see in the picture how the eigenvalues accumulate towards the continuum in correspondence with \eqref{accurate}.

Furthermore, in Fig.~\ref{approxwf} we show an eigenfunction with an appropriately high index and the eigenvalue very close to the continuum. The figure corresponds to the interpolation parameter $b=(2\pi)^{-1/4}$ mentioned in Remark~6.1a above.
\begin{figure}[ht]
\centering
\includegraphics[scale=.75]{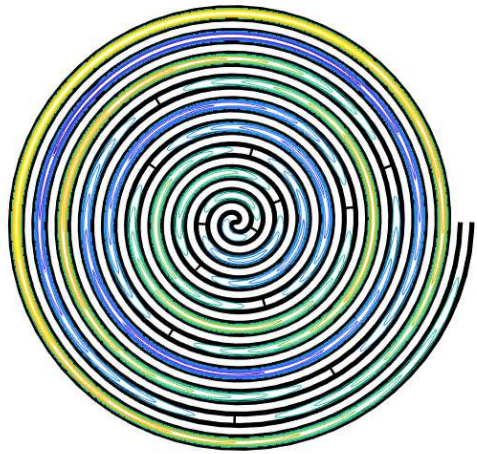}
\caption{The eigenfunction corresponding to the energy $E_{14}=0.999952$ for the interpolation with $b=(2\pi)^{-1/4}$ (colour online).}
\label{approxwf}
\end{figure}
The difference from the two-arm Archimedean region is hardly perceptible by a naked eye, however, the discrete spectrum is now not only non-void but it is rich with the eigenfunctions the tails of which have a distinctively quasi-one-dimensional character.

\section{Concluding remarks}
\setcounter{equation}{0}

The topic we have addressed in this paper is broad and offers a number of other questions. Let us briefly mention some of them:
\begin{itemize}
\setlength{\itemsep}{0pt}
\item There are other spiral-shaped regions of interest. Apart from the example mentioned in Remark~\ref{r:jelly roll} with the mostly mathematical appeal, we discussed situations in which the closure $\overline{\mathcal{C}}$ was simply connected. This is often not the case in physical applications, recall, for instance, \cite{BBP09, CLV14, JLX15, TLY18}, where the guide looks like a two-arm Archimedean region -- in the experimentalist reckoning mentioned above -- however, with an opening in the center where the two `loose ends' meet each other in an S-shape way.
\item With respect to Propositions~\ref{prop: disccavity} and \ref{prop: m-disccavity}, a discrete spectrum can also be created by erasing a part of the Dirichlet boundary away from the center. There is no doubt that a large enough `window' would give rise to bound states, the question is how small it can be to have that effect.
 \item Having in mind that (single-arm) spiral-shaped regions violate the rotational symmetry, one is naturally interested whether the spectrum is simple.
 \item Another question concerns the spectral statistics for strictly shrinking spirals, one would like to know whether they give rise to quantum chaotic systems.
\item We mentioned in Remark~\ref{r:jelly roll} that properties of the Neumann Laplacian may be completely different. On the other hand, one expects that a Robin boundary (with a nonzero parameter) would behave similarly to the Dirichlet one with some natural differences: in the asymptotically Archimedean case the essential spectrum threshold would be not $(2a)^{-1}$ but the appropriate principal eigenvalue of the Robin Laplacian on an interval of length $2\pi$, and in the case of an attractive boundary the operator will no longer be positive.
\item In some physical applications a magnetic field is applied, hence it would be useful to investigate spectral properties of the magnetic Dirichlet Laplacian in spiral-shaped regions of the types discussed here.
\item And last not least, in real physical system the separation of the spiral coils is never complete, which motivates one to look into the `leaky' version of the present problem, that is, singular Schr\"odinger operators of the type $-\Delta+\alpha\delta_\Gamma$.
\end{itemize}

\subsection*{Acknowledgment}

The work of P.E. was in part supported by the European Union within the project CZ.02.1.01/0.0/0.0/16 019/0000778.

\subsection*{ORCID iDs}

Pavel Exner  https://orcid.org/0000-0003-3704-7841


\end{document}